\documentclass[oneside,american]{amsart}
\usepackage[T1]{fontenc}
\usepackage[utf8]{inputenc}
\usepackage{amsthm}
\usepackage{amssymb}

\makeatletter
\numberwithin{equation}{section}
\numberwithin{figure}{section}

\makeatother

\theoremstyle{plain}
\newtheorem{thm}{\protect\theoremname}
\newtheorem{prop}[thm]{\protect\propositionname}
\newtheorem{lem}[thm]{\protect\lemmaname}
\usepackage{babel}
\providecommand{\lemmaname}{Lemma}
\providecommand{\propositionname}{Proposition}
\providecommand{\theoremname}{Theorem}

\begin{document}
\title[A uniformly efficient rejection sampler for MED]{A Uniformly Efficient Rejection Sampler for the Multivariate Expectile-Based
Distribution}
\author{Sam Power}
\begin{abstract}
The multivariate expectile-based distribution is a `spiked' perturbation
of a multivariate Gaussian distribution, proposed in \cite{arbel2023multivariate}.
The rejection sampler proposed for this distribution in the initial
work is valid, but its acceptance probability can degenerate badly
both in high dimension and in the strongly asymmetric limit. We give
a simple alternative. After affine whitening and a polar decomposition,
the sampling problem reduces exactly to a univariate distribution,
and an additional hyperbolic change of variables exhibits this distribution
as log-concave, so that a universal and uniformly efficient construction
of Devroye \cite{devroye2014random} applies. The resulting exact
sampler therefore enjoys an acceptance probability at least $1-\exp\left(-1\right)=0.632120\ldots$,
uniformly over the dimension and all admissible asymmetry parameters. 
\end{abstract}

\maketitle

\section{The MED sampling problem}

Let $d\geq2$, let $\mu\in\mathbb{R}^{d}$, and let $\Sigma\in\mathbb{R}^{d\times d}$
be symmetric and positive definite. For $\left\Vert x\right\Vert _{\Sigma^{-1}}:=\left(x^{\top}\Sigma^{-1}x\right)^{1/2}$,
$\left\langle x,y\right\rangle _{\Sigma^{-1}}:=x^{\top}\Sigma^{-1}y$,
fix $\nu\in\mathbb{R}^{d}$ satisfying $\left\Vert \nu\right\Vert _{\Sigma^{-1}}=1$,
and let $0\leq\rho<1$.

Arbel et al.\ \cite{arbel2023multivariate} introduce the multivariate
expectile-based distribution. We write $X\sim\mathsf{MED}\left(\mu,\Sigma,\nu,\rho\right)$
when $X$ has this probability law on $\mathbb{R}^{d}$ with density
\[
\begin{split}f_{d}\left(x;\mu,\Sigma,\nu,\rho\right) & =\frac{C_{d}\left(\rho\right)}{\left(2\pi\right)^{d/2}\left|\Sigma\right|^{1/2}}\\
 & \qquad\exp\left(-\frac{1}{2}\left\Vert x-\mu\right\Vert _{\Sigma^{-1}}\left(\left\Vert x-\mu\right\Vert _{\Sigma^{-1}}+\rho\left\langle x-\mu,\nu\right\rangle _{\Sigma^{-1}}\right)\right),
\end{split}
\]
where 
\[
C_{d}\left(\rho\right)=\sqrt{1-\rho^{2}}\left(\frac{1+\sqrt{1-\rho^{2}}}{2}\right)^{(d-2)/2}.
\]
When $\rho=0$, this reduces to $\mathcal{N}\left(\mu,\Sigma\right)$.

Choose $A$ such that $\Sigma=AA^{\top}$ and set $v=A^{-1}\nu$.
Then $\left\Vert v\right\Vert =1$, and the affine representation
$X=\mu+AY$ reduces the sampling problem to drawing a standardised
variable $Y$ from a density of the form
\[
f\left(y\right)\propto\exp\left(-\frac{\left\Vert y\right\Vert ^{2}}{2}-\frac{\rho}{2}\left\Vert y\right\Vert \left\langle v,y\right\rangle \right).
\]

\section{The Gaussian rejection sampler}

Arbel et al.\ \cite{arbel2023multivariate} obtain a rejection sampler
from the elementary bound $\left|\left\langle v,y\right\rangle \right|\leq\left\Vert y\right\Vert $,
suggesting the Gaussian upper envelope
\[
\exp\left(-\frac{\left\Vert y\right\Vert ^{2}}{2}-\frac{\rho}{2}\left\Vert y\right\Vert \left\langle v,y\right\rangle \right)\leq\exp\left(-\frac{1-\rho}{2}\left\Vert y\right\Vert ^{2}\right),
\]
and a rejection sampler based on the proposal $\mathcal{N}\left(0,\frac{1}{1-\rho}I_{d}\right)$.
They characterise the acceptance probability of this proposal as
\[
\alpha^{\mathrm{Gauss}}_{d}\left(\rho\right)=\frac{\left(1-\rho\right)^{d/2}}{C_{d}\left(\rho\right)}=\sqrt{\frac{1-\rho}{1+\rho}}\left(\frac{2\left(1-\rho\right)}{1+\sqrt{1-\rho^{2}}}\right)^{(d-2)/2},
\]
which is readily seen to degenerate as $\rho\uparrow1$, particularly
for large $d$. This motivates the search for an improved proposal.

\section{A symmetry-aware representation}

The inefficiency of the Gaussian envelope comes from treating the
MED tilt as a genuinely $d$-dimensional perturbation. In fact, the
perturbation depends only on the radius and one angular coordinate.

For the standardised variable $Y$, define $Q=\frac{\left\Vert Y\right\Vert ^{2}}{2}$,
$\Theta=\frac{Y}{\left\Vert Y\right\Vert }$, $T=\left\langle v,\Theta\right\rangle $.
Under the standard Gaussian reference distribution, $Q$ is distributed
as $\mathsf{Gamma}\left(\frac{d}{2},1\right)$, $\Theta$ is uniform
on $\mathbb{S}^{d-1}$, and $T$ is drawn from the law on $\left[-1,1\right]$
with density given by
\[
h_{d}\left(t\right)=c_{d}\left(1-t^{2}\right)^{(d-3)/2},\qquad c_{d}=\frac{\Gamma\left(\frac{d}{2}\right)}{\sqrt{\pi}\Gamma\left(\frac{d-1}{2}\right)}.
\]
Here and below, the second argument of $\mathsf{Gamma}$ denotes its
rate. Moreover, $Q$ is independent of $\Theta$ and hence of $T$;
of course, $T$ is itself a function of $\Theta$. This coordinate
system allows for a simplified perspective on the MED family.
\begin{prop}
Under the standardised MED law, the marginal density of $T$ is proportional
to
\[
\frac{\left(1-t^{2}\right)^{(d-3)/2}}{\left(1+\rho t\right)^{d/2}}.
\]
Conditional on $T=t$, $Q$ is distributed as $\mathsf{Gamma}\left(\frac{d}{2},1+\rho t\right)$;
and conditional on $\left(Q,T\right)=\left(q,t\right)$, $Y$ is equal
in law to $\sqrt{2q}\left\{ tv+\sqrt{1-t^{2}}W\right\} $, where $W$
is uniform on the unit sphere in $v^{\perp}$.
\end{prop}

\begin{proof}
The density of the MED with respect to the standard Gaussian may be
written as 
\[
\exp\left(-\frac{\rho}{2}\left\Vert Y\right\Vert \left\langle v,Y\right\rangle \right)=\exp\left(-\frac{\rho}{2}\left\Vert Y\right\Vert ^{2}T\right)=\exp\left(-\rho QT\right).
\]
Consequently, the joint density of $\left(Q,T\right)$ is proportional
to 
\[
q^{d/2-1}\exp\left(-q\left(1+\rho t\right)\right)\left(1-t^{2}\right)^{(d-3)/2}.
\]
Since $1+\rho t>0$ for every $t\in[-1,1]$, this is a Gamma density
in $q$ conditional on $t$. Integrating out $q$ gives 
\[
\int^{\infty}_{0}q^{d/2-1}\exp\left(-q\left(1+\rho t\right)\right)\,\mathrm{d}q=\frac{\Gamma\left(\frac{d}{2}\right)}{\left(1+\rho t\right)^{d/2}},
\]
which proves the marginal formula. The reconstruction of $Y$ is the
usual polar decomposition conditional on its projection onto $v$. 
\end{proof}

\section{A hyperbolic change of coordinates}

Introduce the hyperbolic coordinate $\Psi=\tanh^{-1}T\in\mathbb{R}$,
and the parameter $\psi_{0}=\tanh^{-1}\rho\in\mathbb{R}$. Application
of the change-of-variables formula shows that the density for $\Psi$
on $\mathbb{R}$ is proportional to
\begin{align*}
g_{d,\rho}\left(\psi\right) & =\left(\cosh\psi\right)^{-\left(d-1\right)}\left(1+\rho\tanh\psi\right)^{-d/2}\\
 & \propto\left(\cosh\psi\right)^{-(d/2-1)}\left(\cosh\left(\psi+\psi_{0}\right)\right)^{-d/2}.
\end{align*}
Write its negative log-density (up to an unimportant additive constant)
as 
\[
\begin{split}U_{d,\psi_{0}}\left(\psi\right) & :=\left(\frac{d}{2}-1\right)\log\cosh\psi+\frac{d}{2}\log\cosh\left(\psi+\psi_{0}\right).\end{split}
\]

\begin{prop}
For every $d\geq2$ and $\psi_{0}\in\mathbb{R}$, the function $U_{d,\psi_{0}}$
is strictly convex and coercive on $\mathbb{R}$, with derivatives
\begin{align*}
U^{\prime}_{d,\psi_{0}}\left(\psi\right) & =\left(\frac{d}{2}-1\right)\tanh\psi+\frac{d}{2}\tanh\left(\psi+\psi_{0}\right)\\
U^{\prime\prime}_{d,\psi_{0}}\left(\psi\right) & =\left(\frac{d}{2}-1\right)\mathrm{sech}^{2}\psi+\frac{d}{2}\mathrm{sech}^{2}\left(\psi+\psi_{0}\right).
\end{align*}
It has a unique minimiser at $\psi_{\star}=\tanh^{-1}t_{\star}$,
where, writing $\rho=\tanh\psi_{0}$,
\[
t_{\star}=-\frac{d\rho}{d-1+\sqrt{\left(d-1\right)^{2}-d\left(d-2\right)\rho^{2}}}.
\]
\end{prop}

\begin{proof}
The derivative calculations are immediate, and the displayed second
derivative proves strict convexity. To determine the mode, write $t=\tanh\psi$
and use the hyperbolic identity $\tanh\left(\psi+\psi_{0}\right)=\frac{t+\rho}{1+\rho t}$;
solving $U^{\prime}_{d,\psi_{0}}\left(\psi\right)=0$ then reduces
to a quadratic equation in $t$ which can be solved explicitly.
\end{proof}

\section{The Devroye reduction}

We record here a construction of Devroye \cite{devroye2014random}
which gives a universal solution to the task of univariate rejection
sampling in the unnormalised, log-concave setting. The construction
involves forming a rather explicit `three-piece lower envelope' to
the negative log-density. We give the pertinent statement for our
result, leaving the proof and further discussion to the original paper.
\begin{lem}
Let $U:\mathbb{R}\to\mathbb{R}$ be differentiable, strictly convex,
and coercive. Let $\psi_{\star}$ be its global minimiser, and let
$\psi_{-}<\psi_{\star}<\psi_{+}$ solve $U\left(\psi_{-}\right)=U\left(\psi_{+}\right)=U\left(\psi_{\star}\right)+1$.
Let $\underline{U}:\mathbb{R}\to\mathbb{R}$ be the piecewise-affine
minorant to $U$ obtained by drawing tangents at the three points
$\psi_{-},\psi_{\star},\psi_{+}$ and then taking their pointwise
maximum. Then $U\geq\underline{U}$ globally, and it holds that
\[
\frac{\int_{\mathbb{R}}\exp\left(-U\left(\psi\right)\right)\,\mathrm{d}\psi}{\int_{\mathbb{R}}\exp\left(-\underline{U}\left(\psi\right)\right)\,\mathrm{d}\psi}\geq1-\exp\left(-1\right).
\]
Consequently, rejection sampling with envelope $\exp\left(-\underline{U}\right)$
and target $\exp\left(-U\right)$ has acceptance probability at least
$1-\exp\left(-1\right)$.
\end{lem}

\section{Uniform efficiency}

Fix $d\geq2$ and  $0\leq\rho<1$. Given  $\mu\in\mathbb{R}^{d}$,
symmetric positive-definite $\Sigma$, and $\nu\in\mathbb{R}^{d}$
with $\left\Vert \nu\right\Vert _{\Sigma^{-1}}=1$, consider the following
procedure.
\begin{enumerate}
\item Choose $A$ such that $\Sigma=AA^{\top}$, and set $v=A^{-1}\nu$. 
\item Set $\psi_{0}=\tanh^{-1}\rho$, construct the three-piece envelope
for $U_{d,\psi_{0}}$, and use it to draw $\Psi$; set $T=\tanh\Psi$.
\item Draw $Q\sim\mathsf{Gamma}\left(\frac{d}{2},1+\rho T\right)$.
\item Draw $G\sim\mathcal{N}\left(0,I_{d}\right)$, and set $W=\frac{G-\left\langle v,G\right\rangle v}{\left\Vert G-\left\langle v,G\right\rangle v\right\Vert }$. 
\item Form $Y=\sqrt{2Q}\left\{ Tv+\sqrt{1-T^{2}}W\right\} $, and return
$X=\mu+AY$.
\end{enumerate}
\begin{thm}
The preceding procedure gives an exact draw from $\mathsf{MED}\left(\mu,\Sigma,\nu,\rho\right)$.
Its only rejection step has acceptance probability satisfying 
\[
\alpha^{\mathrm{MED}}_{d}\left(\rho\right)\geq1-\exp\left(-1\right)=0.632\ldots
\]
uniformly over $d\geq2$ and $0\leq\rho<1$. Consequently, 
\[
\sup_{d\geq2,\;0\leq\rho<1}\mathbf{E}\left[N\right]\leq\frac{e}{e-1}=1.58\ldots.
\]
\end{thm}

\begin{proof}
The radial--angular proposition shows that an exact MED draw is obtained
once $T$ is drawn from its stated marginal distribution, followed
by the conditional Gamma and uniform spherical draws. The hyperbolic
proposition shows that the law of $\Psi=\tanh^{-1}T$ is log-concave,
whereby the three-piece envelope lemma applies and facilitates a uniformly
efficient rejection sampler. All subsequent steps are direct conditional
draws, and the expected proposal bound is the reciprocal of the acceptance
lower bound. 
\end{proof}

\section{Comments and outlook}
\begin{enumerate}
\item The original Gaussian rejection sampler is attractive because it requires
no preprocessing beyond a Gaussian draw. However, treating the MED
asymmetry as a fully $d$-dimensional perturbation leads to a sampler
which is not uniformly efficient. The sampler proposed herein isolates
the actual non-Gaussian component of the problem, and exposes the
essentially univariate complexity of the problem.
\item The construction requires neither the MED normalising constant nor
any evaluation of special functions. The mode is explicit, while the
abscissae $\psi_{\pm}$ are available as the solutions of two monotone
scalar equations, reliably accessible by e.g., bisection search. Alternatively,
a second-order Taylor approximation might suggest using the approximate
abscissae $\widetilde{\psi}_{\pm}=\psi_{\star}\pm\left(2U^{\prime\prime}\left(\psi_{\star}\right)^{-1}\right)^{1/2}$
to initialise an adaptive rejection sampling routine, which one expects
to be rather efficient in practice. Some omitted calculations show
that even this cruder approach has acceptance probability at least
$0.375\ldots$ uniformly over the dimension and asymmetry parameter.
Numerical minimisation suggests that its actual worst-case acceptance
probability is approximately $0.828\ldots$, indicating very reliable
performance in practice.
\end{enumerate}

\section{Acknowledgements}

The substantive use of AI tools in this work is localised to the final
comments about the adaptive rejection sampling routine and its initialisation,
where conversations with OpenAI ChatGPT convinced the author that
the Taylor-based initialisation was viable and could enjoy an analogous
uniform efficiency. General encouragement from Julyan Arbel in the
early stages of this project is also warmly appreciated.

\bibliographystyle{plain}
\bibliography{med}

\end{document}